\documentclass[a4paper,10pt]{article}

\usepackage{url}
\usepackage{amsthm, amsmath,amssymb,amsfonts,graphicx,algorithmic,algorithm,caption,subfigure}
\usepackage[a4paper, textwidth=124mm, textheight=190mm]{geometry}
\usepackage{breakcites}
\usepackage{sectsty}

\newcommand{\mysubsubsection}[1]{\subsubsection*{\textbf{#1}}}
\newcommand{\REM}[1]{}
\newcommand{\trans}[1]{\,{\stackrel{{#1}}{\rightarrow}}\,}
\newcommand{\transs}[1]{\,{\stackrel{{#1}}{\twoheadrightarrow}}\,}

\newtheoremstyle{mytheoremstyle} % name
    {\topsep}                    % Space above
    {\topsep}                    % Space below
    {}                   % Body font
    {}                           % Indent amount
    {\bfseries \fontsize{10}{15}}                   % Theorem head font
    {.}                          % Punctuation after theorem head
    {.5em}                       % Space after theorem head
    {}  % Theorem head spec (can be left empty, meaning ‘normal’)
\theoremstyle{mytheoremstyle}

\newtheorem{theorem}{Theorem}
\newtheorem{proposition}{Proposition}
\newtheorem{definition}{Definition}
\newtheorem{remark}{Remark}
\newtheorem{notation}[theorem]{Notation}
\newtheorem{example}{Example}
\newtheorem{observation}{Observation}

\newenvironment{keywords}{
      \list{}{\advance\topsep by0.35cm\relax\small
       \leftmargin=1cm
       \labelwidth=0.35cm
       \listparindent=0.35cm
      \itemindent\listparindent
      \rightmargin\leftmargin}\item[\hskip\labelsep
                                     \bfseries Keywords:]}
     {\endlist}

\renewenvironment{proof}
	{\begin{trivlist}\item[\hskip\labelsep{\em Proof.}]}
	{\leavevmode\unskip\nobreak\quad\hspace*{\fill}{\ensuremath{{\square}}}\end{trivlist}}

\sectionfont{\fontsize{12}{15}\selectfont}
\subsectionfont{\fontsize{11}{15}\selectfont}
\date{}
\graphicspath{{./Images/}}

\begin{document}

%%%%%%%%%%%   Heading   %%%%%%%%%%%

\title{\bfseries\fontsize{12}{15}\selectfont  A Process Algebraic Form to Represent Extensive Games \footnote{published in control and cybernetics journal, no. 1, vol. 44, 2015.}}

\author{ {\bfseries \fontsize{10}{15}\selectfont  Omid Gheibi}\\
   {\fontsize{10}{15}\selectfont  (CE Department, Sharif University of Technology, Iran}\\
   {\fontsize{10}{15}\selectfont  gheibi@ce.sharif.edu)}\\
\smallskip\smallskip
{\bfseries  \fontsize{10}{15}\selectfont  Rasoul Ramezanian}\\
   {\fontsize{10}{15}\selectfont  (Complex and Multi Agent System Lab} \\ 
   {\fontsize{10}{15}\selectfont  Mathematical Sciences Department, Sharif University of Technology, Iran}\\
{\fontsize{10}{15}\selectfont     ramezanian@sharif.edu)}
}

\maketitle

%%%%%%%%%%%   Abstract   %%%%%%%%%%%

\begin{abstract} 
\noindent In this paper, we introduce an agent-based representation
of games, in order to propose
a compact representation for multi-party games in game theory.  Our method is  inspired by concepts in process theory and process algebra. 
In addition, we introduce an algorithm whose input is a game in the form of process algebra (proposed in this paper) and as an output, finds the Nash equilibrium of the game in linear space complexity.

\end{abstract}


%%%%%%%%%%%   Introduction   %%%%%%%%%%%

\section{Introduction}
Extensive representation form of a game is a directed graph whose nodes are players, and edges are actions. Assuming that the game has $n$ agents, and each agent has two actions available, the game can be represented by a graph of size $O(2^n)$. 

However, by explaining  the behavior of each agent individually using  an adequate process (called \emph{process-game}), and  obtaining the whole  game through \emph{parallel composition}  these \emph{process-game}s, it is possible to represent the same game in
$O(n)$  space. We take advantage of process algebra to define process-game and the appropriate notion of parallel composition for them.

The word ``process'' refers to the behavior of a system which is a set of actions that are performed in the system and the order of their executions. 
The process theory makes it possible to model the behavior of a system with enormous complexity through modeling the behavior of its components~\cite{MM05}.
By taking advantages of process algebra--automating calculations and running algorithms using parallel computing techniques--we can code the process theory terms and definitions~\cite{W07},~\cite{TADJ}.
On the other hand, a game is a system and its behavior is established by the behavior of all of its players (components). 
Hence, the game could be studied and formally modeled as interactive process~\cite{BENT}. 

Moving along this path in formal methods, in order to reduce the representation of  games with lots of players, we modify the process theory in an appropriate manner to provide a model called ``\textit{process-game}" that encompasses both process theory and game theory notions.  
This proposed process algebraic model  makes it possible to have a compact representation for extensive games--specially in social extensive games which have local interaction--via appropriate parallel composition (section 3).

Also, there are other efforts to reduce the representation of games with very high number of players. 
In comparison with graph-based representation that is proposed to reach the same goal~\cite{KLS}, our  proposed model, facilitates the reduction of  representation of games in the formal method. 

Eventually, to manipulate process-game model efficiently, we propose an algorithm to find the equilibrium path of  games in linear space complexity by using a revision of depth first search and backward induction (section 4).

%%%%%%%%%%%   Preliminaries  %%%%%%%%%%%

\section{Preliminaries}
In this section, we review some preliminary concepts in game and process theory.

\subsection{Game-Theoretic Concepts}\label{gameConcepts}
 \noindent In this part we briefly review definitions and concepts of strategic and extensive games with perfect information which appeared in the literature using the same notations as in~\cite{MR94}, Page 89. 
 
A strategic game is a model of decision-making such that decision-makers choose their plan simultaneously  from their possible actions, once and for all.

\begin{definition}[Strategic Game] A strategic game consists of:
\begin{itemize}
\item A finite set of players $N$
\item for each player $i\in N$ a nonempty set of actions $A_i$
\item for each player $i\in N$ a payoff function $\Pi_i: A_1\times \cdots \times A_n \rightarrow \mathcal{R}$.
\end{itemize}
\end{definition}

\noindent 
If for each player $i$'s action set $A_i$ is finite, the the game is finite.

\begin{example}[prisoner's dilemma]\label{ex:dilemma}
Two members of a criminal gang are arrested and imprisoned. 
Each prisoner is in solitary confinement with no means of speaking to or exchanging messages with the other. 
The police admit they don't have enough evidence to convict the pair on the principal charge. 
They plan to sentence both to a year in prison on a lesser charge. 
Simultaneously, the police offer each prisoner a Faustian bargain. 
Each prisoner is given the opportunity either to betray ($B$) the other, by testifying that the other committed the crime, or to cooperate ($C$) with the other by remaining silent. 
Here's how it goes:
\begin{itemize}
\item If A and B both betray the other, each of them serves 2 years in prison
\item If A betrays B but B remains silent, A will be set free and B will serve 3 years in prison (and vice versa)
\item If A and B both remain silent, both of them will only serve 1 year in prison (on the lesser charge)
\end{itemize}
The above situation is shown as a strategic game in Figure~\ref{fig:dilemma}.

\begin{figure}[h]
\begin{center}
\begin{tabular}{ r|c|c| }
\multicolumn{1}{r}{}
 &  \multicolumn{1}{c}{B}
 & \multicolumn{1}{c}{C} \\
\cline{2-3}
B~ &  \hspace*{0.25cm} (-2,-2) \hspace*{0.25cm} & \hspace*{0.25cm} (0,-3) \hspace*{0.25cm} \\
\cline{2-3}
C~ & (-3,0) & (-1,-1) \\
\cline{2-3}
\end{tabular}
\end{center}
\caption{\small strategic representation of Example~\ref{ex:dilemma}\label{fig:dilemma}}
\end{figure}

\noindent
The set of players is $N=\{1,2\}$. 
Possible actions for each player come from the set of $A_1 = A_2 = \{B,C\}$. 
Each cell in the above table corresponds to a  \emph{strategy profile} and shows the resulting payoffs. 
A strategy profile is a set of strategies for all players which fully specifies all actions in a game. 
A strategy profile must include one and only one strategy for every player.
In the cell, the first value is the player 1's payoff and the player 2's payoff is the second one.
\end{example}

One of the most common solution concepts in game theory is Nash equilibrium.
This notion captures a steady state which no player wants to deviate from the current state if action of the other players are fixed (therefore, all player choose their action in a rational manner).

\begin{notation}
For each strategy profile $s$,  $s_{-i}$ shows the action of all players, except player $i$.
\end{notation}

\begin{definition}[Nash Equilibrium] A Nash equilibrium of a strategic game $\Gamma = \left< N,(A_i), (\Pi_i)\right>$ is a strategy profile $s^*$ with the property that for every player $i \in N$ we have
\[
\Pi_i(s^*_{-i},s^*_i) \geq \Pi_i(s^*_{-i},s_i), \forall s_i \in A_i.
\]
\end{definition}

\begin{example} In Example~\ref{ex:dilemma}, strategy profile $(B,B)$ is a Nash equilibrium. 
Because, If player 1 deviates from action $B$ to $C$, his payoff  decreases from $-2$ to $-3$ in strategy profile $(C,B)$. Therefore, he has not any motivation to deviate from the state $(B,B)$. 
Also, because of symmetry, we can say the same reason for player 2 to have not any motivation to deviate from the current state. 
Therefore, strategy profile $(B,B)$ is a Nash equilibrium in this game.
\end{example}

Another form of games, which is the pivot of discussion in this paper, is extensive game. In this form, decision-makers act sequentially (unlike strategic game).

\begin{definition}[Extensive Game with Perfect Information] \label{def: EGwPI} An extensive game with perfect information is defined by a four-tuple $\left<N,H,P,(\Pi_i)\right>$ which has the following properties:
\begin{itemize}
  \item A set $N$ of \emph{players}
  \item A set $H$ of sequences (finite or infinite) that satisfies the following three properties:
  \begin{itemize}
    \item The empty sequence $\emptyset$ (the \emph{empty history} representing the start of the game) is a member of $H$.
    \item If $(a^{k})_{k=1,\ldots,K}\in H$ (where K may be infinite) and positive integer $L<K$ then $(a^{k})_{k=1,\ldots,L}\in H$.
    \item If an infinite sequence $(a^{k})_{k=1,\ldots,\infty}$ satisfies $(a^{k})_{k=1,\ldots,L}\in H$ for every positive integer $L$ then $(a^{k})_{k=1,\ldots,\infty}\in H$.
  \end{itemize}
  Each member of H is a \emph{history} and each term of a history is an \emph{action} which is taken by a player. A history $(a^{k})_{k=1,\ldots,K}\in H$ is \emph{terminal} if it is infinite or there is no $a^{K+1}$ such that $(a^{k})_{k=1,\ldots,K+1}\in H$. The set of all terminal histories is denoted by $Z$.
  \item A function $P$ (the \emph{player function}) that assigns to each nonterminal history (each member of $H\backslash Z$) a member of $N$, $P(h)$ returns the player who takes an action after the history $h$ ($P:H\backslash Z\rightarrow N$).
  \item A function $\Pi$ (the \emph{payoff function}) that assigns to each terminal history (each member of $Z$) a member of $\mathbb{R}^{|N|}$($\Pi:Z\rightarrow \mathbb{R}^{|N|}$, $\Pi_{i}(z)$ is player $i$'s payoff in terminal history $z \in Z$).
\end{itemize}
\end{definition}

An extensive game with perfect information is \emph{finite} if and only if the set $H$ of possible histories is finite.
Throughout this paper, whenever we use the term extensive games, we mean extensive games with perfect information.
In an extensive game, $P(h)$ chooses an action after any nonterminal history $h$ from the set $A(h)=\{a:(h,a) \in H\}$ where $(h,a)$ means a history $h$ followed by an action $a$ which is one of the actions available to the player who moves after $h$.

\begin{definition}[Strategy]\label{strategy} A strategy of player $i \in N$ in an extensive game $\left<N,H,P,(\Pi_i)\right>$ is a function that assigns an action from $A(h)$ to each $h \in H \backslash Z$ (nonterminal history) for which $P(h)=i$.
\end{definition}

The outcome of a strategy profile $s$ is a terminal history which is constructed by $s$, and is denoted by $O(s)$.

\begin{example}\label{huswif} Two people (``husband" and ``wife") are buying items for a dinner party. The
husband buys either fish (F) or meat (M) for the meal; the wife
buys either red wine (R) or white wine (W). Both people prefer red
wine with meat and white wine with fish, rather than either of the
opposite combinations. However, the husband prefers meat over
fish, while the wife prefers fish over meat. 
Assume that the husband buys the meal and tells his wife what was
bought; his wife then buys some wine.
If we want to consider this problem as an extensive game with perfect information we can determine its component like
\begin{itemize}
\item	$N=\{\mathit{Wife},\mathit{Husband}\}$
\item	Possible actions for husband are a member of the set $A_{\mathit{Husband}} = \{F,M\}$ and wife's actions come from $A_{\mathit{Wife}}=\{R,W\}$.
So in this example, $Z$ is a set of sequences which are started by the action $F$ or $M$ and terminated by $R$ or $W$. 
All possible histories $H$ and terminal histories $Z$ are shown below:
\[
H=\{(\emptyset), (F), (M),  (F,R), (F,W), (M,R), (M,W)\}
\]
\[
Z=\{(F,R), (F,W), (M,R), (M,W)\}
\]
\item	For each $h\in H\backslash Z$ , $P(h)$ is as follows:
\[
P((\emptyset)) = \mathit{Husband}, P((F)) = \mathit(Wife),P((M)) = \mathit(Wife)
\]
\item We can represent
the preferences as utility-based payoffs:

$$\Pi_\mathit{Husband}(M,R) = 2,~ \Pi_\mathit{Husband}(F,W) = 1,$$
$$ \Pi_\mathit{Husband}(F,R) = \Pi_\mathit{Husband}(M,W) = 0$$
$$\Pi_\mathit{Wife}(M,R) = 1,~ \Pi_\mathit{Wife}(F,W) = 2,~ \Pi_\mathit{Wife}(F,R) = \Pi_\mathit{Wife}(M,W) = 0$$
\end{itemize}
\end{example}

Nash equilibrium is a common solution concept for extensive games. 
This concept is defined in the following definition.

\begin{definition} The strategy profile $s^{\ast}$ in an extensive game with perfect information is a Nash equilibrium if for each player $i \in N$ we have:
\begin{center}
    $\Pi_{i}(O(s^{\ast}))\geq \Pi_{i}(O(s_{i},s^{\ast}_{-i}))$ for every strategy $s_{i}$ of player $i$.
\end{center}
\end{definition}

The notion of a Nash equilibrium ignores the sequential structure of an extensive form game
and treats strategies as if they were choices made once and for all.
To refine the Nash equilibrium definition is defined a new notion of subgame perfect equilibrium  (see ~\cite{NarY014}, Chapter 3).

\begin{definition} Let $\Gamma=\left<N,H,P, (\Pi_i)\right>$ be an extensive game with perfect information. The subgame of $\Gamma$ that follows the history $h \in H\backslash Z$ (a nonterminal history) is the extensive game $\Gamma(h)=\left<N,H|_{h},P|_{h},\Pi|_{h}\right>$ where all sequences $h'$ of actions for which $(h,h')\in H$ are in the set $H|_{h}$ and vice versa, $P|_{h}$ is the same as function $P$ but its domain come from the set $H|_{h}$, and for each $h',h'' \in Z|_{h}$ (the set $Z|_{h}$ consists of all the terminal histories in the set $H|_{h}$) is defined that $\Pi_{i}|_{h}(h')\geqslant\Pi_{i}|_{h}(h'')$ if and only if $\Pi_{i}(h,h')\geqslant\Pi_{i}(h,h'')$ (note that $(h,h'),(h,h'')\in Z$).

\end{definition}

In a moment, we can define the notion of subgame perfect equilibrium as mentioned before, using the notion of subgame.

\begin{definition}
The strategy profile $s^{\ast}$ in an extensive game with perfect information is a subgame perfect equilibrium if for every player $i\in N$ and every nonterminal history $h\in H\backslash Z$ for which $P(h)=i$ we have:
\begin{center}
    $\Pi_{i}|_{h}(O_{h}(s^{\ast}|_{h}))\geq \Pi_{i}|_{h}(O_{h}(s_{i},s^{\ast}_{-i}|_{h}))$ for every strategy $s_{i}$ of player $i$ \\ in the subgame $\Gamma(h)$.
\end{center}
\end{definition}

Most often, an extensive game is shown by a tree and marked the subgame perfect Nash equilibria (SPNE) of the game on the tree as in~\cite{NarY014}, Chapter 3 (for a better understanding see Example~\ref{huswif} and Figure~\ref{HW}).

\begin{figure}[ht]
\centerline{\includegraphics[width=4cm]{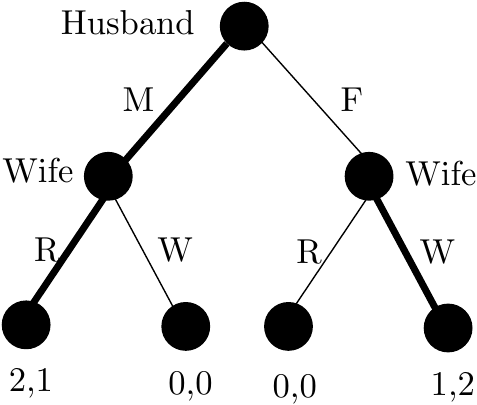}}
\caption{\label{HW}\small game of Example~\ref{huswif} which is depicted as a tree and shown the SPNEs by fat links}
\end{figure}

The definition of  extensive game with perfect information (Definition~\ref{def: EGwPI}) can be generalized by allowing simultaneous moves of the players.
This type of games is called \emph{extensive game with imperfect information}, in this paper.
An extensive game with imperfect information is determined by a quintuple $\left<N,H,P, (\mathcal{I}_i),(\Pi_i)\right>$.
Relative to the definition of an extensive game with perfect information, the new element is the collection $(\mathcal{I}_i)_{i\in N}$ of information partitions.
For each player $i  \in N$,  $\mathcal{I}_i$  is a partition   of $\{h \in H : P(h) = i\}$ with the property that $A(h) = A(h')$ whenever $h$ and $h'$  are in the same member of the partition. 
For $I_i \in \mathcal{I}_i$, we denote by $A(I_i)$ the set  $A(h)$ and by $P(I_i)$, the player $P(h)$ for 
any $h \in I_i$ ($\mathcal{I}_i$ is the \emph{information partition} of player $i$; a set $I_i \in \mathcal{I}_i$ is an \emph{information set} of player $i$)~\cite{MR94}.

We define a variant of the battle of sexes (BoS) game (Example~\ref{huswif}) as an example of extensive games with imperfect information.
\begin{example}\label{varBoS} Assume in Example~\ref{huswif}, that wife decides to hold a dinner party or not. 
If she decides not to hold the dinner party, the game ends and nothing happens.
On the other hand, there are games similar to Example~\ref{huswif} where players move simultaneously instead of moving sequentially. 
As shown in Figure~\ref{SHW}, simultaneous moving is specified by dashed line and means the wife does not know 
that she is in which history.

\begin{figure}[ht]
\centerline{\includegraphics[width=4cm]{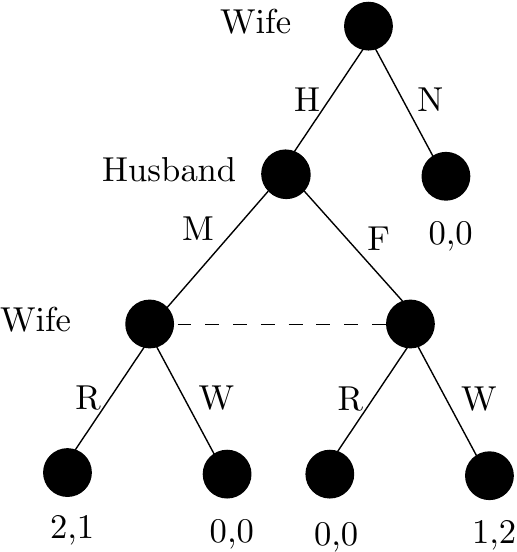}}
\caption{\label{SHW}\small game of the Example~\ref{varBoS} which is depicted as a tree and dashed line means wife and husband move simultaneously on that level (as an example of extensive game with imperfect information).}
\end{figure}
\end{example}
Formally, we have $P(\emptyset) = P(H, M) = P(H, F) = \text{Wife}$, $P(H) = \text{Husband}$, $\mathcal{I}_{\text{Wife}}=\{\{\emptyset\},\{(H, M),  (H, F)\}\}$, and $\mathcal{I}_{Husband} = \{\{H\}\}$.

\begin{definition}
A pure strategy of player $i \in N$  in an extensive game with imperfect information $\left<N,H,P, (\mathcal{I}_i),(\Pi_i)\right>$ is a function that assigns an action in $A(I_i)$ to each information set $I_i \in \mathcal{I}_i$.
\end{definition}

%%%%%%%%%%%   Process Theory   %%%%%%%%%%%

\subsection{Process Theory}\label{ts}

The notion of a transition system can be considered to be
the fundamental notion for the description of process
behavior~\cite{MM05}. In this section we state some abstract
formal definitions  regarding transition systems and specify the notion
of authentication using
 these definitions.

\begin{definition} A transition system $T$ is a quintuple
 $(S,A,\rightarrow,\downarrow,s_0)$ where
 \begin{itemize}
 \item $S$ is a set of states,

 \item $A$ is a set of actions containing an internal action $\tau$,

 \item $\rightarrow\subseteq S\times  A \times S$ is a set of
 transitions,

 \item $\downarrow\subseteq S$ is a set of successfully
 terminating states,

 \item $s_0\in S$ is the initial state.
 \end{itemize}
\end{definition}

 The set $\twoheadrightarrow\subseteq
S\times A^\star\times S$ shows \emph{generalized transitions} of
$T$ ($A^\star$ is a set of all possible chains of actions from $A$). A state $s\in S$ is called \emph{reachable} state of $T$ if
there is $\sigma\in A^\star$ such that $s_0\transs{\sigma}s$. The
set of all reachable states of a transition system $T$ is denoted
by $reach(T)$. We define $act(T)=\{a\in A\mid \exists s,s'\in reach(T)~(s,a,s')\in \rightarrow\}$. 
In the sequel, we assume that
every transition system $T$ is connected, i.e., $reach(T)=S$, and
$act(T)=A$. If $S$ and $A$ are finite, $T$ is called a finite
transition system.

\begin{notation} We refer to $(s,a,s')\in \rightarrow$ by
$s\trans{a}s'$.
\end{notation}

We define $Trace=\{\sigma\in A^\star\mid\exists s'\in S~
s_0\transs{\sigma}s'\}$. If there exists $n\in \mathbb{N}$ such
that $\forall \sigma\in Trace~ (|\sigma|\leq n)$, where $|\sigma|$
is the length of the sequence $\sigma$,  then $T$ is called a
finite-depth transition system. If for every $s\in S$,
$\{(s,a,s')\in \rightarrow\mid a\in A,s'\in S\}$ is finite, then
$T$ is called a finite-branching transition system. By the
notation $\tau$, we refer to the \emph{silent action}.

\begin{proposition}\label{konig} If $T$ is   both a finite-depth and
a finite-branching transition system, then it is a finite
transition system.
\end{proposition}
\begin{proof}   It is straightforward. \end{proof}

\begin{definition} Let $T=(S,A,\rightarrow,\downarrow,s_0)$ be a
transition system. Then $T$ is deterministic if the following
condition holds: Whenever $s_0\transs{\sigma}s$ and
$s_0\transs{\sigma}s'$, then $s=s'$.
\end{definition}

\begin{definition}\label{comf} 
Let $A$ be a set of actions. A
  communication function on $A$ is a partial function $\gamma:
A\times A\rightarrow A$ such that for any  $a,b\in A$:
$\gamma(\tau,a)$ is not defined, and  if $\gamma(a,b)$ is defined
then $\gamma(b,a)$ is defined and $\gamma(a,b)=\gamma(b,a)$.    
The image of $\gamma$ is shown by $C_\gamma$.
We define $H_\gamma=A-C_\gamma$. 
Assume that if $\gamma(a,b)$ is defined then both $a,b\in H_\gamma$.
\end{definition}

\begin{definition}[Parallel Composition]
Let $T=(S,A,\rightarrow,\downarrow,s_0)$ and
$T'=(S',A',\rightarrow',\downarrow',s'_0)$ be two transition
systems, and $\gamma$ a communication function on a set of actions
that includes $A\cup A'$. The parallel composition of $T$ and $T'$
under $\gamma$, written $T\parallel T'$, is the transition system
$(S'',A'',\rightarrow'',\downarrow'',s''_0)$ where \begin{itemize}
\item $S''=S\times S'$,

\item $A''=A\cup   A'\cup \{\gamma(a,a')\mid a\in A,a'\in A'\}$

\item $\rightarrow''$ is the smallest subset of $S''\times
A''\times S''$ such that:
\begin{itemize}
\item[-] if $s_1\trans{a} s_2$ and $s'\in S'$, then
$(s_1,s')\trans{a}''(s_2,s')$, \item[-] if $s'_1\trans{b} s'_2$
and $s \in S $, then $(s ,s_1')\trans{b}''(s ,s_2')$,

\item[-] if $s_1\trans{a} s_2$, $s'_1\trans{b} s'_2$, and
$\gamma(a,b)$ is defined, then
$(s_1,s'_1)\trans{\gamma(a,b)}''(s_2,s'_2)$,
\end{itemize}
\item $\downarrow''=\downarrow\times\downarrow'$,

\item $s''_0=(s_0,s'_0)$.
\end{itemize}
\end{definition}

\begin{definition}[Encapsulation]
Let $T=(S,A,\rightarrow,\downarrow,s_0)$ be a transition system.
Let $H$ be a set of actions. The encapsulation of $T$ with respect
to $H$, written as $\delta_H(T)$, is the transition system
$(S',A',\rightarrow',\downarrow',s'_0)$ where \begin{itemize}
\item $S'=S$, $A'=A$, $\downarrow'=\downarrow$, $s'_0=s_0$ and
\item $\rightarrow'=\rightarrow\cap (S\times(A-H)\times S)$.
\end{itemize}
\end{definition}

Assume $T_1$ and $T_2$ are two processes , and execute them in
parallel.
Then for~$H=A-C_\gamma$,  the encapsulation of the
process $T_1 \parallel T_2$ makes the processes to communicate.
It means the difference between $T_1 \parallel T_2$ and
$\delta_{H}(T_1\parallel T_2)$ is that in the second process only
communication actions exist.

\begin{proposition}\label{depth} If $T$ and $T'$ are two transition systems, and
$T$ is finite-depth, then $\delta_{H_\gamma}(T\parallel T')$ is
finite-depth.
\end{proposition}

\begin{proof} It is straightforward. \end{proof}

\begin{definition} 
Assume two transition systems $T=(S,A,\rightarrow,\downarrow,s_0)$ and
$T'=(S',A',\rightarrow',\downarrow',s'_0)$, and for   $s\in S$,
let $l(s)=\{(s,a,t)\in \rightarrow\mid a\in A,t\in S\}$, and for
every $s'\in S'$, $l' (s')=\{(s',b,t')\in \rightarrow'\mid b\in
A',t'\in S'\}$. We say $T,T'$ are \emph{communication
finite-branching} with respect to communication function $\gamma$,
if for any $(s,s')\in S\times S'$ the set
$\{((s\trans{a}t),(s'\trans{b}'t))\in l(s)\times l(s')\mid
\gamma(a,b)~is~defined\}$ is finite.
\end{definition}

\begin{proposition} If two transition systems $T=(S,A,\rightarrow,\downarrow,s_0)$ and
$T'=(S',A',\rightarrow',\downarrow',s'_0)$ are \emph{communication
finite-branching} with respect to a communication function
$\gamma$, then $\delta_{H_\gamma}(T\parallel T')$ is
finite-branching.
\end{proposition}

\begin{proof} It is straightforward.\end{proof}

%%%%%%%%%%%   The Processes of Games   %%%%%%%%%%%

\section{The Processes of Games}\label{model}
In this section, we introduce a \emph{process-algebraic representation} for extensive games with perfect information in order to reduce the social large games into logarithmic (or polylog) size proportional to the size of \emph{extensive representation} of the same games.

To analyse a game, we need to illustrate it with the game tree. 
If the number of agents is too large then  we may need to use a machinery to provide the game tree illustration.
We propose a machinery for this aim and call it process-game. 
The process game is a combination of process theory and game theory for solving games with large number of agents. 
We model the operation of each agent using a process-game. 
Then, we run all of these process-games in parallel to obtain the game tree.

\begin{definition}Let $A$ be a set of actions.
  The language  $L(A)$, is the smallest superset
of $A^\star$ such that
\begin{center}
	$\rho,\sigma\in L(A)\Rightarrow \neg
	\rho, (\rho\vee\sigma), mid(\sigma), pre(\sigma), pos(\sigma) \in
	L(A)$.
\end{center}
Let $T=(S,A,\rightarrow,\downarrow,s_0)$ be a transition system.
For $s\in S$, the history of $s$, denoted by $h(s)$, is the trace
$\sigma\in A^\star $ such that $s_0\transs{\sigma}s$. For a state
$(T,s)$ and a formula $\sigma\in L(A)$, we define the satisfaction
$(T,s)\models \sigma$, as follows:
\begin{itemize}
	\item[] for $\sigma\in A^\star $, $(T,s)\models \sigma$ iff
	$h(s)=\sigma$,

	\item[] $(T,s)\models pre(\sigma)$ iff $\exists\rho\in A^\star
	(h(s)=\sigma\rho)$,

	\item[] $(T,s)\models mid(\sigma)$ iff $\exists\rho,\varsigma\in
	A^\star (h(s)=\varsigma\sigma\rho)$,

	\item[]  $(T,s)\models pos(\sigma)$ iff $\exists\rho\in A^\star
	(h(s)=\rho\sigma)$,

	\item[]  $(T,s)\models (\rho\vee\sigma)$ iff $(T,s)\models
	 \rho$ or $(T,s)\models \sigma$,

	\item[]  $(T,s)\models \neg\rho$ iff $(T,s)\not\models
	 \rho$.
\end{itemize}

\end{definition}

\begin{definition}
Let $A$ and $B$  be two sets of actions. The set of
conditional actions over $A$ with conditions in $B$, denoted by
$A_{con(B)}$ is defined
  as follows:

\begin{itemize}
\item[] for each $a\in A$, and  $\sigma\in L(B)$, $[\sigma]a$ is a
conditional action, i.e., $[\sigma]a\in A_{con(B)}$.
\end{itemize}
There is an injective mapping from $A$ to $A_{con(B)}$ which maps
each $a\in A$ to $[\top]a\in A_{con(B)}$.

\end{definition}

\begin{definition}[Encapsulation of Conditional Actions]
Let $A$ and $B$ be two sets of actions and
$T=(S,A_{con(B)},\rightarrow,\downarrow,s_0)$ be a transition
system over conditional actions $A_{con(B)}$. The encapsulation of
conditional actions of $T$, written $\delta^c(T)$, is the
transition system $(S',A',\rightarrow',\downarrow',s'_0)$ where
\begin{itemize} \item $S'=S$, $A'=A$, $\downarrow'=\downarrow$,
$s'_0=s_0$ and \item for $s,t\in S'$ and $a\in A$, $s\trans{a}'t$
iff for some $\sigma\in L(B)$, $s\trans{[\sigma]a} t$ and
$(T,s)\models \sigma$.
\end{itemize}

\end{definition}
Now we give an example to illustrate the above definitions.
\begin{example}\label{HWProcess}Consider Example~\ref{huswif}. 
We may consider two processes one for the husband and one for the
wife\begin{itemize}\item[]$husband:=M+F$ (Figure \ref{HP}), and
\item[]$wife:=[M]R+[M]W+[F]R+[F]W$ (Figure \ref{WP}). \end{itemize}Finally, the
transition system of the process $\delta^c(husband\parallel wife)$
is exactly the tree of the game (Figure 2). 
Process algebraic form of the tree of the game is $M.(R+W)+F.(R+W)$.
\end{example}
\begin{figure}
        \centering
        ~ 
        \subfigure[husband process]{
                \centering
                \includegraphics[scale=1]{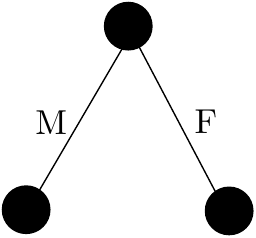}
                \label{HP}
        }
        \subfigure[wife process]{
                \centering
                \includegraphics[scale=1]{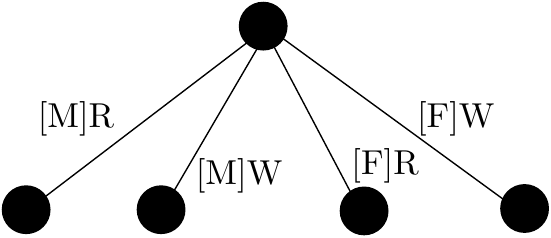}
                \label{WP}
        }
                \caption{processes of Example~\ref{HWProcess}}\label{Processes}
\end{figure}

\begin{definition}
Let $T=(S,A,\rightarrow,\downarrow,s_0)$. The transition system
$\delta_\cup(T)$ is a transition systems obtained from $T$ by
cutting some of its transitions in the following way:
\begin{center}
$s\trans{a} t$ and $b\varsubsetneq a$ then  $s\not\trans{b}$
\end{center}

\end{definition}

Assume we are given a $\Gamma=\left<N,H,P,(\Pi_i)\right>$ which is an extensive game with perfect information. 
Now we model it using process theory as a process-game by mapping each player (in $N$) to a process,
each terminal state to a member of $\downarrow$, and payoff function ($\Pi$) to profit value for each member of $\downarrow$.
\begin{definition} [Process-game] 
Let  $\Gamma=\left<N,H,P,(\Pi_i)\right>$ be an extensive game with perfect information. A process-game model for $\Gamma$ is a tuple  $\mathfrak{P}=\left< (T_i)_{i\in N}, (\pi_i)_{i\in N}\right>$ where each $T_i=(S^i,A^i_{con(B)},\rightarrow^i,\downarrow^i,s^i_0)$ is a transition systems and each $\pi_i: A^\star\rightarrow \mathbb{R}$ is a profit function.
$A_1, A_2,\ldots, A_n$ are conditioned by $B$ ($A^1_{con(B)},\ldots,A^n_{con(B)}$) so that
\begin{itemize}
\item[1.] if $i\neq j$ then $A_i\cap A_j=\emptyset$,

\item[2.] $B=(\bigcup_{i\in N}  A_i)\cup A_1\times A_2 \times \cdots \times A_n$.
\end{itemize}
and the game $\Gamma$ is mapped to process-game $\mathfrak{P}$ so that
 \begin{itemize}
 \item $S^i$ is a set of states where at each state player $i \in N$, decides to perform one of his/her possible actions which is determined by $P$ for each node on game tree,

 \item $A^i$ is a set of actions containing an internal action $\tau$ that represents possible actions of player $i$,

 \item $\rightarrow^i \subseteq S^i\times  A^i_{con(B)} \times S^i$ is a set of transitions that represents what happens when a player chooses one of his actions to do (using $P$),

 \item $\downarrow^i\subseteq S^i$ is a set of successfully terminating states that represents terminal states on game tree which can be defined by terminal histories ($O(s)$),
\item $\pi_i:A^\star \rightarrow \mathbb{R}$ is a payoff function that represents payoff ($\Pi_i$) for each action of the player $i$ in each \emph{subprocess} (like a subgame) which is started by $i$.
 \item $s_0^i\in S^i$ is set of initial states which player $i$ can choose to start his game from.
 \end{itemize}
\end{definition}
Now we can construct the process tree of the game using $\delta_\cup(\delta^c(T_1\parallel T_2\parallel...\parallel T_n))$.

The sizes of $A_i$  and $N$ (set of players or agents) in $T_i$ and $\pi_i$ are denoted by $|A_i|$, $n$, and $|\pi_i|$ respectively. Let $d=\max_i(|A_i|)$ which is called the branching factor of the process/game tree.

\begin{theorem}
Suppose a process-game $\mathfrak{P}=\left< (T_i)_{i\in N},(\pi_i)_{i\in N}\right>$ with $T_i=(S^i, A^i_{con(B)},   \rightarrow^i, \downarrow^i, s^i_0)$ is given. 
The size of the equivalent extensive representation of the $\mathfrak{P}$  which is denoted by $|ERT|$ is $O\left(d^{n+1}\right)$.
\end{theorem}

\begin{proof}
The size of  the extensive representation is equal to the size of the tree of the game. 
As players act sequentially, the maximum number nodes in the first level of the tree is $|A_1|$, in the second level would be $|A_1|\times |A_2|$ and so on. 
Therefore, we have
\begin{center}
$$|ERT| \leq |A_1| + |A_1|\times |A_2| + \cdots + |A_1|\times |A_2| \times \cdots \times |A_n|$$
$$\leq d + d^2 +\cdots + d^{n} = O\left(d^{n+1}\right).$$
\end{center}
\end{proof}

\begin{theorem}\label{th: size of process-game}
The size of a given process-game $\mathfrak{P}=\left< (T_i)_{i\in N},(\pi_i)_{i\in N}\right>$ with $T_i=(S^i,A^i_{con(B)},\rightarrow^i,\downarrow^i,s^i_0)$ is 
$O(nd|B| + \sum_{i=1}^{n}|\pi_i|)$.
\end{theorem}
\begin{proof}
The size of $\mathfrak{P}$ is equal to the sum of $|T_i|$ and $|\pi_i|$ for all players.
The size of $T_i$ is $O(|A^i_{con(B)}|)$. Therefore,
\begin{equation}\label{eq:tSize}
\sum_{i=1}^{n}|T_i| \leq \sum_{i=1}^{n}|A^i_{con(B)}| \leq \sum_{i=1}^{n}|A_i|\times|B| \leq nd|B|,
\end{equation}
\begin{equation}\label{eq:process-game size}
|\mathfrak{P}| = \sum_{i=1}^{n}|T_i| + \sum_{i=1}^{n}|\pi_i|.
\end{equation}
\noindent
We can conclude from the equations~\ref{eq:tSize} and~\ref{eq:process-game size} that
\begin{center}
$|\mathfrak{P}| = O(nd|B|+ \sum_{i=1}^{n}|\pi_i|)$.
\end{center}

\end{proof}

Assume that each action of players are in the form of $a$ or $[b]a$ (without condition or just with one condition)
, therefore	 the size of $B$ for each player is $O(d)$.
According to Theorem~\ref{th: size of process-game}, it is easy to see that the size of the process-game representation would be  $O(nd^2+ \sum_{i=1}^{n}|\pi_i|)$.
Based on the assumption, the size of the process-game can be logarithmic proportional to the size of the extensive representation (as $d$ is the maximum number of plsyers' actions, it is a constant). 
The following equation shows this fact.
\[
|\mathfrak{P}| = O(nd^2+ \sum_{i=1}^{n}|\pi_i|) = O(\log(d^n)\frac{d^2}{\log(d)}+ \sum_{i=1}^{n}|\pi_i|) = O(\log(|ERT|) + \sum_{i=1}^{n}|\pi_i|)\text{.}
\]

%%%%%%%%%%%   Extended Version   %%%%%%%%%%%

\subsection*{Extended Version}
We can extend the definition of process-game for extensive games with imperfect information in the following manner.

\begin{notation} Let $A_1$ and $A_2$ be two disjoint sets. For $(a,b)\in (A_1\cup 2^{A_1}) \times (A_2\cup 2^{A_2})$, we
define $a\dot{\cup}b:= ~$

\begin{itemize}
	\item[] $\{a\}\cup\{b\}$ ~~~if~~$(a,b)\in A_1 \times A_2$,
	\item[] $\{a\}\cup b$ ~~if ~$(a,b)\in A_1 \times 2^{A_2}$,
	\item[] $\{b\}\cup a$ ~~if ~$(a,b)\in 2^{A_1} \times A_2$,
	\item[] $a\cup b$ ~~if ~$(a,b)\in 2^{A_1} \times 2^{A_2}$.
\end{itemize}
\end{notation}
Using the above notation, we modify definition \ref{comf} over $n$ disjoint sets of actions in the following.
\begin{definition}\label{ecomf}
Let $A_1, A_2, \ldots, A_n$ be $n$ disjoint sets of actions. $\gamma: (A_1\cup 2^{A_1}) \times (A_2\cup 2^{A_2}) \times \cdots \times (A_n\cup 2^{A_n}) \rightarrow 2^{A_1\cup A_2\cup \cdots \cup A_n}$ is a communication function which is defined $\gamma(a_1,a_2,\ldots, a_n):=a_1\dot{\cup}a_2\dot{\cup}\cdots\dot{\cup}a_n$ for all $(a_1,a_2,\ldots, a_n)\in (A_1\cup 2^{A_1}) \times (A_2\cup 2^{A_2}) \times \cdots \times (A_n\cup 2^{A_n})$.
\end{definition}
We may extend $\gamma$ to the set of conditional actions in the manner as in the following definition:
\begin{definition}\label{ccomf} Let $A_1,A_2,B$ be three disjoint sets of actions. $\gamma$ is a communication function over conditional actions $A^1_{con(B)}$ and $A^2_{con(B)}$ for $([\sigma]a,[\sigma]b)\in A^1_{con(B)} \times A^2_{con(B)}$ which is defined as follows:
\begin{center}
$\gamma([\sigma]a,[\rho]b):=[\sigma\wedge\rho]\gamma(a,b)$.
\end{center}
\end{definition}
 This communication function helps to model simultaneous players' moves in extensive games with imperfect information
  (See Example~\ref{varBoS}).
\begin{example} The process of each player in Example~\ref{varBoS} is as given below:
\begin{itemize}
\item[] $wife:= (H+N).(R+W)$,
\item[] $husband:=[H]M+[H]F$.
\end{itemize}
and $\gamma$ is defined over $A_w=\{R,W\}$, $A^h_{con(B)}=\{[H]M,[H]F\}$, and $B=\{H\}$. $\gamma([H]M,R)$, for instance, is equal to $[H]\{M,R\}$.
The transition system of the process over $\gamma$ communication function is $\delta^c(wife\parallel husband)$
which is exactly the tree of the game. The tree is shown in Figure~\ref{SHWinP} which is equivalent to Figure~\ref{SHW}.
\end{example}
\begin{figure}[ht]
\centerline{\includegraphics{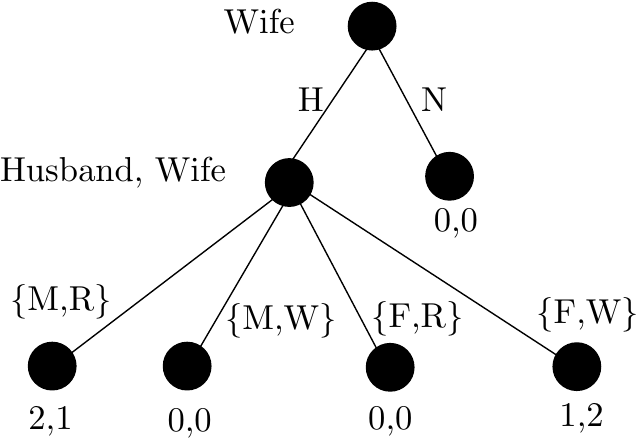}}
\caption{\label{SHWinP}\small This figure shows simultaneous moving of wife and husband in the second level of the game and equivalent to Figure~\ref{SHW}. }
\end{figure}

Now, we can define the \emph{extended process-game} by a tuple $\mathfrak{P}=\left< (T_i)_{i\in N},(\pi_i)_{i\in N}, \gamma \right>$ for a given extensive game with imperfect information $\Gamma = \left<N,H,P, (\mathcal{I}_i), (\Pi_i)\right>$.
Relative to the definition of a process-game, the new element is the $\gamma$ as a communication function over $A_1, A_2,\ldots, A_n$ which are conditioned by $B$. 
Definition of $\gamma$ is based on $\mathcal{I}_i$.

In the extended version, the process/game tree is constructed with the same method as in  process-game ($\delta_\cup(\delta^c(T_1\parallel T_2\parallel...\parallel T_n))$).
Also, if the size of $\gamma$ is denoted by $|\gamma|$, we have 
\begin{equation}
|\mathfrak{P}| = \sum_{i=1}^{n}|T_i| + \sum_{i=1}^{n}|\pi_i| + |\gamma|. \label{eq: extVSize}
\end{equation}
Similar theorems for extended process-game can be proved by replacing equation~\ref{eq: extVSize} in those theorems.

%%%%%%%%%%%   Applications   %%%%%%%%%%%

\subsection*{Applications}
One of the application of process-game in representing the real social extensive games in log or polylog size could be in modeling of social systems 
where agents have \emph{local interaction} and \emph{local competition} (these two terms might be used interchangeably,  conform to this article).
A significant approach which works on the assumption of local interaction is generative social science.

Generative social science~\cite{JE06} is an approach to study social systems via agent-based computational models.
Its aim is to answer that how the regularity of the society emerges from local interaction of heterogeneous autonomous agents.
One of the main assumptions of generative social science is called ``local interaction'' (see~\cite{JE06}, Page 6).
Base on this assumption, agents interact with their neighbors and compete in the social network, in which they are involved.

On the other hand, social extensive games--each player is viewed as being involved in a local competition with the players in geographically neighboring regions--can be modeled as a graph $G$~\cite{KLS}. 
In the graph-based model, each player is represented by a vertex in $G$. The interpretation of edges is that each player is in a game only with the respective neighbors in $G$.

In the above ven, it can be proposed that lots of social systems are based on local interactions and competitions.
The observation forwarded below stipulates that in considering local interaction in the form of extensive games, the size of the process-game can be logarithmic in comparison with the  extensive representation. 

\begin{observation}
Suppose there is an extensive game with perfect information and interaction graph-based model $G$.
The maximum degree of  $G$ is denoted by $\Delta$. 
Actions of each player are limited just to the player's neighbors in $G$ (the size of $B$ in the equivalent process-game would be $O(d^{\Delta})$). 
Hence, the size of the process-game representation would be  $O(nd^{\Delta+1}+ \sum_{i=1}^{n}|\pi_i|)$.
As a consequence, we have
$$|\mathfrak{P}| = O(nd^{\Delta + 1}+ \sum_{i=1}^{n}|\pi_i|) = O(\log(d^n)\frac{d^{\Delta +1}}{\log(d)}+ \sum_{i=1}^{n}|\pi_i|) = $$
$$ O(\frac{d^{\Delta +1}}{\log(d)}\log(|ERT|) + \sum_{i=1}^{n}|\pi_i|)\text{.}$$
As $d,\Delta\ll n$  is a common situation in social extensive games, the size of the process-game can be logarithmic (or polylog, depending on $d$  and $\Delta$) with respect to the size of the extensive representation. 
For instance, considering local competition property, we have  $\Delta =  O(\log(n))$, then $\frac{d^{\Delta +1}}{\log(d)} = O(n)=O(\log(|ERT|))$. Therefore, $|\mathfrak{P}| = O(\log^2(|ERT|) + \sum_{i=1}^{n}|\pi_i|)$.
\end{observation}
\noindent
Note that  local competition is kind of local interaction.
Recalling the assumption of local interaction in generative social science, we can come to the conclusion that 
it is obvious that in lots of social systems $\Delta=  O(\log(n))$.

There have been various representations of extensive games aiming to represent the large games compactly, such as the graphical model~\cite{KLS} and MAIDs \cite{KLM03}.
However, our proposal (originating from the process theory) of producing a compact representation, is quite different.
Also, we bring the advantages of process theory (in execution and management) to the game theory environment.

Our model can be applied in management of complex systems too, but how? Each process has its own profit and each profit is a function of some inputs. 
Suppose a process game is defined with some initial values and may deliver some equilibria path as a result (using the algorithm of the following section). 
However, for a manager, it would be critical to control the path as he wants it to be. 
Therefore he can try changing the inputs (so the profits will be changed, leading also to some new equilibria path) up to getting the best one.

%%%%%%%%%%%   The Algorithm   %%%%%%%%%%%

\section{The Algorithm}

The notion of process-game has been explained completely in the previous section. 
As strategies and agents in the process-game are the same as in its equivalent extensive game, therefore the subgame perfect equilibrium (or equilibrium path) is a solution concept for the process-game, too.
Now, in this section, we propose an algorithm to find the equilibrium path in a process-game.
Algorithm~\ref{preAlgorithm} illustrates the finding of equilibria path under the  assumption that there exists only one equilibrium path for each \emph{subprocess} (like a subgame).% without loss of generality.

\begin{remark} Equilibrium path in a process-game $\mathfrak{P}=\langle (T_i)_{i\in N},(\pi_i)_{i\in N}\rangle$ is a sequence of actions from $s_0^1$ (suppose player $1$ starts the entirely of the game) to $\downarrow^n$ (suppose the last action is for player $n$).
\end{remark}

\renewcommand{\algorithmicrequire}{\textbf{Input:}}
\renewcommand{\algorithmicensure}{\textbf{Output:}}
\renewcommand{\algorithmiccomment}[1]{$\vartriangleright$ #1}
\begin{algorithm}[ht]
\caption{Depth-First Finding Equilibria \label{preAlgorithm}}
\begin{algorithmic}[1]
\REQUIRE A Process-Game $\mathfrak{P}=\langle (T_i)_{i\in N},(\pi_i)_{i\in N}\rangle$
\ENSURE Equilibria path $\Lambda$
%\\
%\COMMENT $s_{j}^i$ means action $j$th of player $i$ and because of depth-first, $i$ is started from $n$ to 1
\FOR{ $s_{j}^i \leftarrow$ each state visited in depth-first expansion of $\mathfrak{P}$ from $s_{0}^1$ using $A^i_{con(B)}$}\label{expansion}
\STATE $isVisited[s_j^i]\leftarrow false$
\STATE $ratPath[s_j^i]\leftarrow\varnothing$
\IF{ $s_{j}^i \in \downarrow^{n} $ }
\STATE $isVisited[s_j^i]\leftarrow true$
\ELSE
\IF{$\forall s',a: (s_j^i,a,s') \in \rightarrow^{i}  \wedge ~isVisited[s']$  }

\STATE $a\leftarrow\arg\max_{a}\pi_i(a.ratPath[s']) $ \COMMENT $s_j^i \trans{a} s' \in \rightarrow^i$
\\
\COMMENT choose a possible action $a\in A_{con(B)}^i$ from state $s_j^i \in S^i$ which maximizes profit \\ \hspace{3mm} of player $i$ in the state $s_j^i$

\STATE $ratPath[s_j^i] \leftarrow a.ratPath[s']$
\STATE $isVisited[s_j^i] \leftarrow true$
\\
\COMMENT the assumption is just one equilibria path exists for each \emph{subprocess}
\STATE delete the $ratPath$ for all $s'$.\label{deletion}
\\
\COMMENT by deleting $ratPath$ for $s'$ the space complexity will remain linear and $ratPath$ \\ \hspace{3mm} propagated toward $s_0^1$ states.
\ENDIF
\ENDIF
\ENDFOR
\RETURN $ratPath[s_0^1]$
\end{algorithmic}
\end{algorithm}

In line \ref{expansion} of Algorithm~\ref{preAlgorithm}, expansion takes place by virtue of conditional actions. 
 As in Algorithm~\ref{preAlgorithm}, each player's payoff value is calculated bottom-up, it is sufficient to save players' payoff in the subprocess equilibria at each level. 
To reuse space and keep the space required by the algorithm linear, we delete all process nodes which are expanded in that subprocess previously in line \ref{deletion} of Algorithm~\ref{preAlgorithm}. 
We present below an example, meant to clarify  how this algorithm works.

\begin{example}\label{algorithm steps}
Two people select a policy that affects them both by alternately vetoing the available (voted) policies until only one remains. 
First person 1 vetoes a policy. 
If more than one policy remains, person 2 then vetoes a policy. 
If more than one policy still remains, person 1 then vetoes another policy. 
The process continues until only one policy has not been vetoed. 
Suppose there are three possible policies, $X$, $Y$, and $Z$, person 1 prefers $X$ to $Y$ to $Z$, and person 2 prefers $Z$ to $Y$ to $X$~\cite{O04}.
Now, we want to represent this situation through the process-game, as a compact representation proportional to the extensive model. 
To define a process-game $\mathfrak{P}=\langle (T_i)_{i\in \{1,2\}},(\pi_i)_{i\in \{1,2\}}, \gamma \rangle$, 
we should first determine the transition system $T_i$.
Players' process model is specified below, with $P1$ and $P2$ being the processes of person 1 and person 2, respectively.
$T_1$ and $T_2$  can be obtained by decoding theses process models.
\begin{itemize}
\item[] $P1 := X + Y + Z$
\item[] $P2 := [X]Y+[X]Z+[Y]X+[Y]Z+[Z]X+[Z]Y$
\end{itemize}

\noindent
Definitions of $\pi_1$ and $\pi_2$  are based on the players' preferences.
The payoff function of each player on each subprocess is equal to the player's payoff  over the complete path from the root to the leaf, passing through the subprocess.
This path for each subprocess $p$ is denoted by  $path(p)$. Therefore,
$$\pi_1 (p) = \left\{ 
\begin{array}{ll}
2 & \text{ if there is no } X \text{ in the } path(p) \\
1 & \text{ if there is no } Y \text{ in the } path(p) \\
0 & \text{ if there is no } Z \text{ in the } path(p)
\end{array}
\right.
,$$
$$ \pi_2 (p) = \left\{ 
\begin{array}{ll}
2 & \text{ if there is no } Z \text{ in the } path(p) \\
1 & \text{ if there is no } Y \text{ in the } path(p) \\
0 & \text{ if there is no } X \text{ in the } path(p)
\end{array}
\right..
$$

\noindent
Actually, we define the above functions compactly. 
Thus, their space complexity is much lower than when defining the function for each path separately.
Now, let us compute the Nash equilibrium. 
We know that the process tree is constructed completely by $\delta_\cup(\delta^c(P1\parallel P2))$.
However, using Algorithm~\ref{preAlgorithm}, the process tree is expanded step by step to save the space.  
The steps of the DFS expansion of the process tree are sketched in Figure~\ref{FNIPG}.

\begin{figure}[ht]
\centerline{\includegraphics[width=11cm]{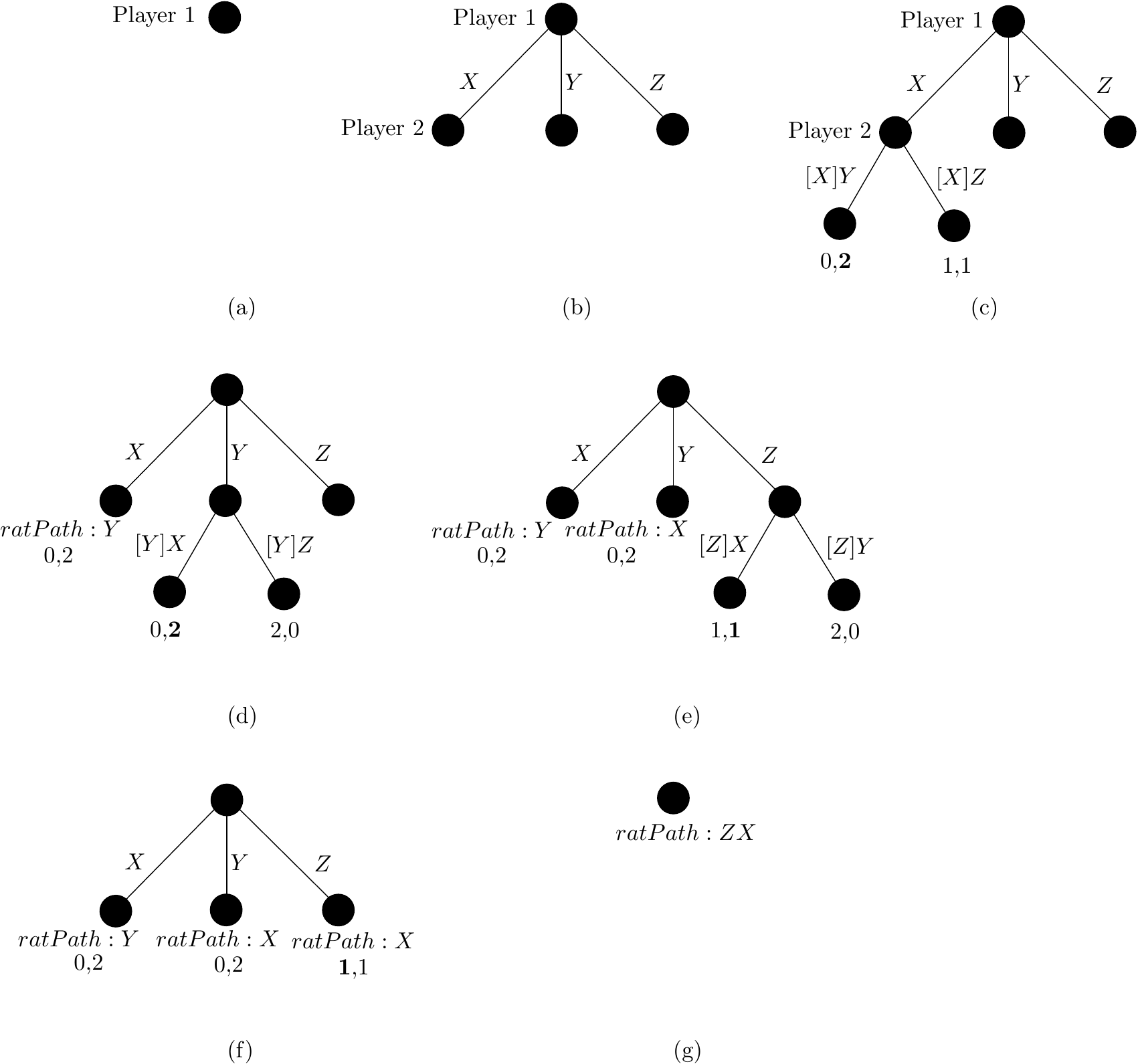}}
\caption{\label{FNIPG}\small  (a) to (g) show the steps of Algorithm~\ref{preAlgorithm} in finding the Nash equilibrium for the process model of Example~\ref{algorithm steps}. 
At step (g) the rational path of the $s^1_0$ ($ZX$) is the Nash equilibrium.}
\end{figure}

\end{example}

 As it is not necessary existing any pure equilibrium path for an extensive game with imperfect information, this algorithm is not working for the extended process-game. However, we can detect the situation in the bottom-up calculation and report ``no pure equilibrium path".

%%%%%%%%%%%   Complexity   %%%%%%%%%%%

\subsection*{Complexity}
Time complexity of Algorithm \ref{preAlgorithm} in worst case would be in the NP-complete complexity class, like for backward induction~\cite{NRTV07} (because we want to find pure Nash equilibria). 
Its space complexity is linear in the size of the game which is given as an input i.e., linear in the depth of the process-game, maximum number of actions which are possible to do by a player, and the size of the payoff function.

In the extended version, if the number of simultaneous moves grows, the extensive game will be transforming to the pure strategic form, so that in the worst case, its space complexity would be exponential.
\\
The algorithm is like a depth-first search and the space complexity of depth-first search is $O(hd)$, 
here $h$ is the height of the tree and $d$ is the branching factor that is the maximum number of actions which a player can do. 
However, there is a bottleneck which is caused by the size of the payoff function. 
If it has exponential size, as space complexity is linear with respect to the size of the payoff function too, it will be exponential too. 
Therefore, the algorithm will be better than backward induction or using strategic form algorithms in terms of space complexity under two circumstances. 
First, the size of the payoff function is polynomial on $n$ and $d$. 
Second, the size of simultaneous moves (required size to represent function of $\gamma$) is polynomial in $n$ and $d$ (in the case of extended version).

\begin{theorem}
For a given extended process-game $\mathfrak{P}=\langle (T_i)_{i\in N},(\pi_i)_{i\in N}, \gamma \rangle$  as an input, the space complexity of Algorithm~\ref{preAlgorithm} is $O(nd + |\mathfrak{P}|)$.
\end{theorem}

\begin{proof}
At each step of Algorithm~\ref{preAlgorithm}, one process node is expanded and finally collapsed when all its subnodes are visited.
Therefore, when Algorithm~\ref{preAlgorithm} is running, at all levels of the tree, at most one node is expanded. 
We know that the height of the tree is at most $n$ and the branching factor of the tree is $d$.
Therefore, the allocated space for expanding the process tree during the running in all steps would be $O(nd)$ in the worst case. 
Hence, the space complexity of the algorithm is $O(nd + \text{size of input}) = O(nd + |\mathfrak{P}|)$.
\end{proof}

Actually, the process-game is a simple case of the extended process-game with $|\gamma| = 0$.
Hence, the above theorem is in a general case true for the process-game, too.

%%%%%%%%%%%   Conclusions and Future Works   %%%%%%%%%%%

\section{Conclusions and Future Works}
We introduced a new model to represent large extensive games  in  a compact representation (specially social extensive games with local competition).
In addition, the model is defined in algebraic terms and can be ran in parallel mode.
Further, we provide an algorithm to find the Nash equilibrium for  this representation in linear space complexity, with respect to the size of the input.
\\
As we mentioned, one of the applications of the model can be in management of complex systems. 
However, there is no software to facilitate the management process for	 the manager.
Therefore, the next step of the work may be develop a software to approach this particular goal.

%%%%%%%%%%%   Acknowledgment   %%%%%%%%%%%

\mysubsubsection{Acknowledgment}
We would like to thank anonymous referees for their helpful comments and suggestions which resulted in a better representation of our paper.

\end{document}